\documentclass[12pt]{article}
\usepackage{pudlatex}
\usepackage{fullpage}

\title{On the complexity of finding falsifying assignments for
  Herbrand disjunctions}

\author{Pavel Pudl\'ak
\thanks{The author is supported by the ERC Advanced Grant
  339691 (FEALORA) and the institute grant RVO: 67985840}}

\begin{document}
\maketitle

\begin{abstract}
  Suppose that $\Phi$ is a consistent sentence. Then there is no
  Herbrand proof of $\neg \Phi$, which means that any Herbrand
  disjunction made from the prenex form of $\neg \Phi$ is
  falsifiable. We show that the problem of finding such a falsifying
  assignment is hard in the following sense. For every total
  polynomial search problem $R$, there exists a consistent $\Phi$ such
  that finding solutions to $R$ can be reduced to finding a falsifying
  assignment to an Herbrand disjunction made from $\neg \Phi$. It has
  been conjectured that there are no complete total polynomial search
  problems. If this conjecture is true, then for every consistent
  sentence $\Phi$, there exists a consistence sentence $\Psi$, such
  that the search problem associated with $\Psi$ cannot be reduced to
  the search problem associated with $\Phi$.
\end{abstract}


\section{Introduction}

Let $\Phi:=\forall x_1\dots\forall x_k\phi(x_1\dts x_k)$ be a
universal sentence (where $\phi$ is an open formula). According to
Herbrand's theorem, $\Phi$ is inconsistent if and only if, for some terms
$\tau_{ij}$, the disjunction  $\bigvee_{i=1}^n
  \neg\phi(\tau_{i1}\dts \tau_{ik})$ is a propositional
  tautology. Thus if $\Phi$ is consistent, every conjunction  $\bigwedge_{i=1}^n
  \phi(\tau_{i1}\dts \tau_{ik})$ is a satisfiable proposition. In this
  paper we study the computational problem of finding satisfying
  assignments for such conjunctions assuming that $\Phi$ is
  consistent. 
We call this problem the \emph{Herbrand consistency search} for $\Phi$.
This problem can be viewed from three different perspectives:

\medskip
1. We ask how difficult it is to verify that $\Phi$ is consistent;
more precisely, how difficult it is to verify that a given
disjunction is not a Herbrand proof of $\neg \Phi$. This is somewhat similar to
the well-known problem about finitistic consistency statements,
where we ask how difficult it is to find a proof that there is no proof
of $\neg \Phi$ of length $n$, see~\cite{finitistic}. The two problems
are, however, of essentially different nature. For one thing, we
consider all proofs of length $n$ when we talk about finitistic consistency.
For another, transforming the usual proofs (Hilbert style, sequent
calculus with cuts, etc.) into Herbrand proofs results in
nonelementary blowup of size.

2. A model of a consistent sentence can be built on the Herbrand
universe (the set of all terms in the language of $\Phi$). To this end
we have to decide the truth of all atomic formulas so that the
resulting structure is a model of $\Phi$. Herbrand consistency search
can be then viewed as the problem of deciding the truth values of
atomic formulas in order to obtain a \emph{partial model} of $\Phi$.

3. It is important to fully understand the complexity of special cases
of the $\bf NP$-complete problem SAT (satisfiability of CNF
formulas). Each consistent universal sentence whose matrix is a CNF
gives us a natural class of CNFs in the way described above. For some
sentences, we can show that the problem is solvable in polynomial
time. For some other sentences, we believe that this is not the case,
but it is unlikely that we can prove that they are $\bf NP$-hard
problems, because the CNFs are all satisfiable. Instead, we can argue
that for strong sentences their Herbrand consistency search problems
are not solvable in polynomial time, because these problems capture
the complexity of all total polynomial search problems, as we explain
below.

\medskip
A \emph{polynomial search problem} is given by a binary relation $R(x,y)$
decidable in polynomial time and a polynomial bound on
the length of $y$ in terms of the length of $x$. The task is, for a
given $x$, to find $y$ such that the relation holds true and the
polynomial bound is satisfied, if there is
any such $y$. A \emph{total polynomial search problem} is a polynomial
search problem that has a solution for every $x$. While there are
polynomial search problems that are $\bf NP$-hard, it seems unlikely
that we could prove  $\bf NP$-hardness of a total polynomial search problem. The
additional condition of totality prevents us to use any know
techniques for showing $\bf NP$-hardness, which suggest that it may actually be
impossible.

There are naturally defined reductions of one total polynomial search
problem to another. This enables us to study classes of these problems
closed under reductions and a number of important classes have been
defined \cite{papadim}. These classes are useful for classification of
specific search problems. In proof complexity such classes are used to
characterize certain sentences provable in fragments of Bounded
Arithmetic (see, e.g., \cite{BK,skelly-thapen}). The structure of the
quasiorder of polynomial reducibility has not been much studied,
except for specific classes of problems. One can easily show that for
every finite set of total polynomial search problem there is another
one to which all are reducible. Since the condition of totality is not
syntactical, we are not able to prove that there is a greatest element
in this quasiorder, i.e., that there is a complete total polynomial
search problem. We conjecture that there is none.

In this paper we prove that every total polynomial search problem is
reducible to the Herbrand consistency search for some consistent
sentence $\Phi$. This means that Herbrand consistency search problems
can have arbitrary high complexity in the hierarchy of total
polynomial search problems. We also prove that polynomial reducibility
reflects the strength of consistent sentences in the sense that if a
universal sentence $\Psi$ logically follows from $\Phi$, then the
Herbrand consistency search for $\Psi$ is reducible to Herbrand
consistency search for $\Phi$. This is in line of our project to find
connections between provability and computational complexity, see
Section~6.4 of~\cite{logical}. We will also define Herbrand
consistency search for general sentences in prenex form, but the
relation between provability and reducibility of the corresponding
Herbrand consistency search problems is not clear for sentences that
are not universal.

The conjecture that there are no complete total polynomial search
problems can be partially justified by showing an oracle with respect
to which it holds true. Since we have not found this result in the
literature, we present it in the last section.

\section{Preliminaries}

We will consider first order logic without equality, but constants and
function symbols will play an important role. Let $\Sigma:=\exists x_1\dts
x_k\sigma(x_1\dts x_k)$ be an existential sentence (where $x_1\dts
x_k$ are all variables in $\sigma$). \emph{Herbrand's Theorem} states
that $\Sigma$ is provable (logically valid) if and only if there exist
terms $\tau_{ij}$, $i=1\dts n$, for some $n$, $j=1\dts k$ such that 
\[
\bigvee_{i=1}^n\sigma(\tau_{i1}\dts \tau_{ik})
\]
is a propositional tautology (see, e.g., \cite{buss,buss-h}). We will study the
dual version of this statement: a universal sentence
$\Phi:=\forall x_1\dts x_k\phi(x_1\dts x_k)$ is consistent if and only
if for all families of terms $\tau_{ij}$, $i=1\dts n$, $j=1\dts k,$
\[
\bigwedge_{i=1}^n\phi(\tau_{i1}\dts \tau_{ik})
\]
is satisfiable as a propositional formula (i.e., we can assign truth
values to the atomic formulas so that the truth value of the
conjunction is truth). A general sentence in a prenex form can be
transformed into a universal sentence by skolemization, which we
denote by 
\[
Sk(\forall\bar{x}\exists\bar{y}\forall\bar{z}\exists\bar{u}\dots
\phi(\bar{x},\bar{y},\bar{z},\bar{u}\dots))\ :=\ 
\forall\bar{x}\forall\bar{z}\dots
\phi(\bar{x},\bar{f}(\bar{x}),\bar{z},\bar{g}(\bar{x},\bar{z})\dots),
\]
where we use bars to denote strings of symbols and $f,g,\dots$ are new
function symbols. If $\Phi$ is $\bigwedge_i\Phi_i$, where the
sentences $\Phi_i$ are in prenex form, then we define $Sk(\Phi)$ to be
$\bigwedge_iSk(\Phi_i)$ (where each term uses different function
symbols). Clearly, $Sk(\Phi)\vdash\Phi$, but the opposite is not true
in general.  For the sake of simplicity, we will only define Herbrand
consistency search for conjunctions of prenex sentences, although
Herbrand's theorem has been proved for general sentences.

\begin{definition}
  Let $\Phi$ be a consistent sentence which is a conjunction of
  sentences in prenex form. Let $\phi(x_1\dts x_k)$ be the matrix
  (the quantifier-free part) of
  the skolemization of $\Phi$.  Then $HCS(\Phi)$, the \emph{Herbrand
    Consistency Search for $\Phi$}, is the following total polynomial
  search problem: \bi
\item given terms $\tau_{ij}$ in the
  language of $\phi$, $i=1\dts n$, $j=1\dts k$, find a truth
  assignment to the atomic subformulas occurring in $\phi(\tau_{i1}\dts
  \tau_{ik})$, for $i=1\dts n$, that makes $\bigwedge_{i=1}^n
  \phi(\tau_{i1}\dts \tau_{ik})$ true.
\ei
\end{definition}

{\it Example 1.} Consider an axiomatization of the theory of dense
linear orderings. Using a Skolem function $f(x,y)$, we can present it
as a universal theory with the axioms (stated without the universal quantifiers)
\[\begin{array}{l}
0<1,\\ 
\neg x<x,\\ 
x<y\vee x=y\vee y<x,\\ 
x<y\wedge y<z\to x<z,\\ 
x<f(x,y)\wedge f(x,y)<y, 
\end{array}\]
plus the identity and equality axioms. 
Let $\phi(x,y,z)$ be the conjunction of these axioms.  Given terms
$\tau_{i,j}$, $i=1\dts m$, $j=1,2,3$, we can easily (certainly in
polynomial time) find truth assignments to the atomic formulas
$\tau_{i,j}=\tau_{i',j'}$ and $\tau_{i,j}<\tau_{i',j'}$ such that the conjunction
\(\bigwedge_i\phi(\tau_{i,1},\tau_{i,2},\tau_{i,3}) \) becomes
true. To find such an assignment we need only to find an
interpretation of the terms in a finite linear ordering and then to
assign the truth values according to this interpretation. To find such
an interpretation, we start by ordering the variables of the terms
$\tau_{i,j}$ in an arbitrary way. Then we gradually extend the
ordering to more complex subterms of the terms
$\tau_{i,j}$. Specifically, having an interpretation of terms $\tau$
and $\sigma$ and a non-interpreted term $f(\tau,\sigma)$, we place
$f(\tau,\sigma)$ on an arbitrary position strictly between $\tau$ and
$\sigma$.

\bigskip {\it Example 2.} Let $\Phi$ be a prenex form sentence 
axiomatizing a fragment of Peano Arithmetic. Consider a skolemization
of $\Phi$. If $\Phi$ is sufficiently complex, some Skolem functions
may be difficult to compute, or they even may be non-computable. Then
finding interpretation in a finite part of the natural numbers may
also be difficult. Note, however, that this does not imply that
finding a satisfying truth assignment must be difficult. In
particular, finding such an assignment is always doable in
nondeterministic polynomial time whatever the complexity of the Skolem
functions is.

\begin{definition}
A \emph{total polynomial search problem} is defined by a binary
relation $R(x,y)$ computable in polynomial time and a polynomial $p$
such that for every $x$ there exists $y$ such that $|y|\leq p(|x|)$ and
$R(x,y)$. The task is, for a given $x$, to find a $y$ satisfying the
two conditions above.
\end{definition}
Here we use $|x|$ to denote the length of $x$, i.e., the number of
bits in an encoding of $x$. In the following definition we will omit
polynomial bounds on $y$s and assume that they are implicit in $R$ and $S$.

\begin{definition}
Let $R$ and $S$ be total polynomial search problems. We say
that $R$ is \emph{polynomially reducible} to $S$ if $R$ can be solved in
polynomial time using an oracle that gives solutions to $S$. We say
that $R$ is \emph{many-one polynomially reducible} to $S$, if it is
polynomially reducible using one query to the oracle for $S$.
\end{definition}
Clearly, both relations are reflexive and transitive. Note that if
{\bf P}={\bf NP}, then every search problem is reducible to every
other one. Hence we can only prove non-reducibility assuming some
conjectures in computational complexity.

\section{Main result}

\begin{theorem}\label{t-1}
  For every total polynomial search problem $R$, there exist a
  consistent universal sentence $\Phi$ such that the problem $R$ is
  many-one polynomially reducible to $HCS(\Phi)$.
\end{theorem}
\begin{proof}
  Given a total polynomial search problem $R$, the sentence $\Phi$
  will express that $R$ is total. This can, certainly, be done in
  various ways, but it does not automatically guarantee that we can
  reduce $R$ to $HCS(\Phi)$. Therefore we have to describe the
  formalization in more detail.

  We start with a brief high-level overview of the proof.  We will
  take a Turing machine $M$ that decides in polynomial time the
  relation $R$ and express that for a given $x$ there exists $y$ and
  an accepting computation of $M$ on the inputs $x$ and $y$. Thus the
  first step is to define terms that will represent an input word
  $x$. Then we need to ensure that the bits of $x$ are encoded into
  the truth values of some atomic formulas. To this end we use an
  elementary theory of the successor function $S$ and use terms
  (numerals) $S^i(0)$ as indices of a one dimensional
  array. Specifically, we use atomic formulas $P(x,S^i(0))$ to
  determine the bits of $x$ ($P(x,S^i(0))$ false means $x_i=0$,
  $P(x,S^i(0))$ true means $x_i=1$). A
  computation of $M$ can be represented by a two-dimensional array
  with entries in a finite alphabet. The elements of the alphabet can
  be encoded by bit strings of length $d$ for some constant $d$. So we
  represent the computation by $d$ ternary relations $Q_k(z,s,t)$. The
  second part of the input $y$ will be implicitly encoded in the
  array. Given a term $\tau$ representing an input word $x$, the term
  $F(\tau)$, where $F$ is a function symbol, will denote the object
  representing the computation. Thus the bits of the array corresponding to
  $F(\tau)$ will be defined by the truth values of $Q_k(F(\tau),S^i(0),S^j(0))$.
The matrix of $\Phi$ will be a conjunction of
several formulas which we can view as axioms of a simple theory
describing computations of $M$. One of the axioms says that $M$
accepts, so the implicitly encoded $y$ must be such that $R(x,y)$
holds true. It will not be hard to see that we need only a polynomial
number of term instances of the axioms in order to guarantee that the
truth values encode a computation on the input word correctly. In fact
these term instances can easily be defined from the input word. The
implicitly encoded $y$ can also be easily read from the truth values,
thus the construction gives a many-one polynomial reduction.

\medskip
Now we describe the formalization in more detail, but since it is
fairly routine, we leave some parts to the reader.

Let a total polynomial search problem be given by a relation $R$
computable in polynomial time. So we assume that for every $x$ there
exists a $y$ such that $R(x,y)$ and the length of $y$ is bounded by a
polynomial in the length of $x$. 
Let $M$ be a (deterministic) Turing machine that in polynomial time
decides the relation $R(x,y)$. We will also assume that $M$ has a certain
form that will make the formalization easier. Specifically, we will
assume the following properties of $M$.  
\ben
\item For given $x,y\in \{0,1\}^*$, $M$ always stops after $p(|x|)$
  steps, where $p$ is some polynomial, provided that the input word $x$
  is coded appropriately (see below). This means that it reaches one
  of the two final states, one of which is the accepting state and the other
  is the rejecting state. 
\item The tape of $M$ is infinite in one direction. The squares of the
  tape will be indexed by $0,1,2,\dots$. We will view squares as
  having $d$ registers indexed $1 \dts d$; every register contains $0$
  or $1$. The contents of a square encode the symbol on the tape,
  the presence/non-presence of the head and the state of the machine.
\item Registers 1 and 2 will be used to encode $x$. The content of
  registers 1 are the bits of $x$ and registers 2 determine the end of
  the word $x$ (the first $1$ in register 2 is in the first square
  after the end of $x$). The input word $y$ will be coded by registers
  3 and 4 in the same way. An occurrence of 1 in register 5 marks the
  position of the head of the machine.
\item Initially all registers with numbers greater than 5 contain
  zeros. Registers 5 contain only one 1 and this is in the square 0.
\item Register 6 will be used to determine that $M$ has stopped and
  rejected; i.e., if 1 occurs in any of the registers 6, then the
  machine rejects.
\item The machine starts by looking for the mark that determines the
  end of $x$. After that it looks for the mark that determines the end
  of $y$. If it does not find it in the given polynomial limit, it
  will stop and reject. 
If the mark is all right, the machine
  computes the relation $R(x,y)$, i.e., it will stop and accept iff
  the relation holds true.
\een

Our sentence $\Phi$ will use relation symbols $=,P(x,t)$,
$Q_i(z,s,t)$, for $i=1\dts d$, constants $0,\Lambda$, and function symbols
$S(x),f_0(x),f_1(x),\ell(x),F(x)$. The sentence will be a universal closure of
formulas that we present in a form of a finite number of axioms.

First we need 
\ben
\item[1.] the axioms of identity and the axiom of equality for $S$
\[
s=t\to S(s)=S(t).
\]
\een 
We do not postulate the axioms of equality for other function and
relation symbols, since we only need them to derive the inequalities
in Lemma~\ref{lA}. Note that these axioms can be stated using three
variables, say, $r,s,t$. The symbol $S$ represents the successor
function, so we postulate the usual axioms 
\ben
\item[2.] $0\neq S(t)$,\ $s\neq t\to S(s)\neq S(t)$.
\een

We leave the proof of the following easy fact to the reader.
\begin{lemma}\label{lA}
The propositions $S^i(0)\neq S^j(0)$ for all $i,j\leq n$, $i\neq j$ are
  derivable using propositional logic from the term instances of
  axioms 1. and 2. for all terms of the form $S^k(0)$, $k\leq n$.
\end{lemma}

Next we need some axioms in order to be able to write down terms that
represent input words $x$. The intended interpretation of the
predicate $P$ is:  the $i$-th bit of $x$ 
is $0$ if $P(x,S^i(0))$ is false, and $1$ otherwise.
The constant $\Lambda$ represents the empty word and
$\ell(x)$ represents the length of a binary word $x$. Therefore
our first axiom is 
\ben
\item[3.] $\ell(\Lambda)=0$.
\een
The functions $f_0$ and $f_1$ add bits $0$ and $1$ at the end of the
word.
\ben
\item[4.] 
\begin{tabbing}
\= $\ell(f_0(x))=S(\ell(x))\ \wedge$\\
\>$\neg P(f_0(x),\ell(x))\ \wedge$\\
\>$(s\neq \ell(x) \to (P(f_0(x),s)\equiv P(x,s)))$.
\end{tabbing}

\item[5.] 
\begin{tabbing}
\= $\ell(f_1(x))=S(\ell(x))\ \wedge$\\
\>$P(f_1(x),\ell(x))\ \wedge$\\
\>$(s\neq \ell(x) \to (P(f_1(x),s)\equiv P(x,s)))$.
\end{tabbing}
\een

Thus given a word $w=(w_0\dts w_{n-1})\in\{0,1\}^n$, the term
$f_{w_{n-1}}\dots f_{w_1}f_{w_0}(\Lambda)$ represents it in our
theory. We need also to show that this fact has a propositional
proof using a small number of instances of the axioms.

\begin{lemma}\label{lB}
Let $\tau=f_{w_{n-1}}\dots f_{w_1}f_{w_0}(\Lambda)$. The
propositions 
\[
(\neg)^{w_0}P(\tau,0),(\neg)^{w_1}P(\tau,S(0))\dts
(\neg)^{w_{n-1}}P(\tau,S^{n-1}(0)), 
\]
and 
\[
\ell(\tau)\neq 0,\ell(\tau)\neq S(0)\dts
\ell(\tau)\neq S^{n-1}(0),\ell(\tau)=S^n(0), 
\]
are derivable using propositional logic from term instances of axioms
1.-5. for terms $S^k(0)$, $k=0\dts n$, and
$\Lambda,f_{w_0}(\Lambda)\dts f_{w_{n-1}}\dots f_{w_1}f_{w_0}(\Lambda)$.
(We denote by $(\neg)^0$ the empty
symbol, and $(\neg)^1$ stands for $\neg$.)
\end{lemma}
\begin{proof}
By induction construct such proofs for all subterms of $\tau$. The
induction step is done using Lemma~\ref{lA} and axioms 4. and 5.
\end{proof}

We represent a computation of the machine by a two dimensional array
where each entry has $d$ registers, each register containing one bit. The first
index is time, the second is a position on the tape. The content of
the $k$-th register in time $s$ and position $t$ is determined by a
predicate $Q_k(z,s,t)$. The variable $z$ stands for the entire
array. The sentence $\Phi$ will express the fact that, for every $x$, there
exists $y$ such that $M$ accepts the input $(x,y)$. We do not need to
mention $y$ explicitly, because it is encoded in the array $z$. We use
skolemization to eliminate the existential quantifier, thus the array
will be represented by $F(x)$.

The initial configuration of the machine is formalized by the
following axioms.
\ben
\item[6.] $Q_1(F(x),0,t)\equiv P(x,t)$, $Q_2(F(x),0,t)\equiv
  \ell(x)=t$.

(Clearly, the predicate $P(x,t)$ is dispensable and can be replaced
by $Q_1(0,F(x),t)$, but it would complicate the presentation above.)
The second input is encoded in the same way using $Q_3$ and $Q_4$, but
we do not need any axioms about it.

\item[7.] $Q_5(F(x),0,0)$, $\neg Q_5(F(x),0,S(t))$,
\item[8.] $\neg Q_i(F(x),0,t)$, for $i=6\dts d$.
\een

The transition function is formalized by axioms of the form:
\ben
\item[9.] $Q_i(F(x),S(s),0)\equiv \rho_i$,\\ 
  $Q_i(F(x),S(s),S(t))\equiv \psi_i$,\\ 
  for $i=1\dts d$, 
\een
where $\rho_i$ and $\psi_i$ are propositions composed from atomic
formulas of the form $Q_j(F(x),s,0)$,  $Q_j(F(x),s,S(0))$,
repspectively,  $Q_j(F(x),s,t)$, $Q_j(F(x),s,S(t))$, $Q_j(F(x),s,SS(t))$
for $j=1\dts d$. 

Finally, we postulate that the machine never rejects the input:
\ben
\item[10.] $\neg Q_6(F(x),s,t)$.
\een

Let $\phi(x,r,s,t)$ be the conjunction of the axioms 1.-10., and let
$\Phi$ be $\forall x\forall r\forall s\forall t\ \phi(x,r,s,t)$. 

\medskip To show that $\Phi$ satisfies the theorem, we have first to
show that $\Phi$ is consistent. To this end we take a function
$\gamma$ such that $R(x,\gamma(x))$ is true for all $x$. We interpret
the predicate symbols $P(x,t)$, $Q_i(z,s,t)$, for $i=1\dts d$,
constants $0,\Lambda$, and function symbols
$S(x),f_0(x),f_1(x),\ell(x)$ as explained above. The function symbol
$F(x)$ represents the function that maps a given string $x$ to the
array encoding the computation of the machine $M$ on the input
$(x,\gamma(x))$. Hence $\Phi$ is consistent.


\medskip
Second, we have to construct a reduction
from the search problem to finding truth assignments of the term
instances of $\phi$. The reduction 
is defined as
follows. Let $w\in\{0,1\}^n$ be given. 
Let $\tau_l$ denote the term $f_{w_{l}}\dots f_{w_1}f_{w_0}(\Lambda)$
for $l=0\dts n-1$.
Finding a solution $u$ such that
$R(w,u)$ will be reduced to finding a truth assignment to the atomic
formulas of 
\[
\bigwedge_{i,j,k=0}^{p(n)}\bigwedge_{l=0}^{n-1}\phi(F(\tau_l),S^i(0),S^j(0),S^k(0))
\]
that makes this formula true. We will denote this formula by
$\Psi_w$. Note that we need the second conjunction to run over all
numbers $l=0\dts n-1$, because we need to derive formulas from
Lemmas~\ref{lA} and \ref{lB}, but the instances of the axioms 6.-10. for
$x=\tau_l$, $l<n-1$ will not be used.

It is clear that $\Psi_w$ can be constructed in polynomial
time, so we only need to show that from any satisfying assignment
$A$, we can construct some $u$ in polynomial time such that
$R(w,u)$. 
We start by observing that according to Lemma~\ref{lA},
$A(P(\tau_{n-1},S^i(0)))=\top$ iff $w_i=1$ for $i=0\dts
n-1$. Similarly for the atomic formulas $\ell(\tau_{n-1})=S^i(0)$ for
$i=0\dts n$, so the
truth values for these formulas represent $w$. By axioms 6., $w$ is
also correctly represented by the truth values of $Q_1(F(\tau_{n-1}),0,S^i(0))$
and $Q_2(F(\tau_{n-1}),0,S^i(0))$. The axioms 7.-9. then ensure that the
truth values of $Q_k(F(\tau_{n-1}),S^i(0),S^j(0))$, for $i,j=0\dts
p(n),k=1\dts d$ represent a computation of the Turing machine $M$ on
$w$ and some $u$, where $u$ is coded by the truth values of
$Q_3(F(\tau_{n-1}),0,S^i(0))$ and $Q_4(F(\tau_{n-1}),0,S^i(0))$. Since the machine
must stop within the limit $p(n)$ and the instances of the axiom
10. ensure that it does not reject, the computation must be
accepting. Hence the string $u$ is such that $R(w,u)$. We just note
that $u$ can easily be constructed from the truth assignment $A$. 

\end{proof}

\section{Reductions among HCS problems}

In order to show connection between provability of $\Phi\to\Psi$ and
polynomial reducibility of $HCS(\Psi)$ to $HCS(\Phi)$, we prove that
provability implies polynomial reducibility if $\Psi$ is universal. 

\begin{proposition}\label{prop4.1}
  Let $\Phi$ be a consistent sentence in a prenex form. Let $\Psi$ be
  a universal sentence such that $\Phi\vdash\Psi$. Then $HCS(\Psi)$ is
  polynomially reducible to $HCS(\Phi)$.
\end{proposition}
\begin{proof}
  Note that $Sk(\Phi)\vdash\Phi$ and $HCS(Sk(\Phi))$ is the same as
  $HCS(\Phi)$. Hence we can w.l.o.g. assume that $\Phi$ is
  universal. 

Let $\Phi$ and $\Psi$ be $\forall x_1\dots
  x_k\phi(x_1\dts x_k)$ and $\forall y_1\dots y_l\psi(y_1\dts
  y_l)$ and assume that  $\Phi\vdash\Psi$. Then we have 
\[
\vdash \forall y_1\dots y_l\exists x_1\dots x_k
(\phi(x_1\dts x_k)\to\psi(y_1\dts y_l)).
\]
The herbrandization of this sentence is
\[
\vdash \exists x_1\dots x_k
(\phi(x_1\dts x_k)\to\psi(c_1\dts c_l)),
\]
where $c_1\dts c_l$ are new constants. According to Herbrand's
theorem, there exist terms $\tau_{ij}$ such that
\bel{e-x}
\bigvee_i \left(\phi(\tau_{i1}\dts\tau_{ik})\to\psi(c_1\dts c_l)\right)
\ee
is a propositional tautology. 

Let $\sigma_{p1}\dts\sigma_{pl}$, $p=1\dts n$,  be terms in the
language of $\psi$. We substitute these terms into (\ref{e-x}) for
$c_1\dts c_l$. The
resulting formulas propositionally imply the following formula
\[
\bigwedge_p\bigwedge_i\phi(\tau^*_{pi1}\dts\tau^*_{pik})\to
\bigwedge_p\psi(\sigma_{p1}\dts\sigma_{pl}),
\]
where
$\tau^*_{pij}:=\tau_{ij}[\sigma_{p1}/c_1\dts\sigma_{pl}/c_l]$. Thus in
order to satisfy $\bigwedge_p\psi(\sigma_{p1}\dts\sigma_{pl})$, it
suffices to satisfy
$\bigwedge_p\bigwedge_i\phi(\tau^*_{pi1}\dts\tau^*_{pik})$. In
general, the latter formula is not an instance of $HCS(\Phi)$ because
$\Psi$ may use other function symbols. However, note that the role of
terms is only to determine which atomic formulas are same and which
are different. Hence to get an instance of $HCS(\Phi)$ that is
essentially the same propositional formula, it suffices to replace the
maximal terms that are not in the language of $\Phi$ by variables (the
same variables for the same terms, of course).
\end{proof}

We will prove a similar result for existential sentences. 
Note that $HCS(\exists y_1\dts y_m\alpha(y_1\dts y_m))$ is trivial,
because the skolemization of the sentence does not contain any
variables, hence it is a finite problem. Thus one has to state the
result in a slightly different way.

\bl\label{lem4.2} 
Let $\Phi$ be a consistent universal sentence, let
$\alpha(y_1\dts y_m)$ be an open formula with $m$ free variables
and let $c_1\dts c_m$ be constants not occurring in $\Phi$ and
$\alpha$. Then
$HCS(\Phi\wedge\forall\bar{y}(\alpha(\bar{y})\to\alpha(\bar{c})))$ is
polynomially reducible to $HCS(\Phi)$.  
\el

\begin{proof}
Let $\Phi$ be $\forall x_1\dts x_n\phi(x_1\dts x_k)$. Let an instance
of $HCS(\Phi\wedge\forall\bar{y}(\alpha(\bar{y})\to\alpha(\bar{c})))$
be given; i.e., we want to find a satisfying assignment for 
\bel{e2}
\bigwedge_{i=1}^n
  (\phi(\tau_{i1}\dts
  \tau_{ik})\wedge(\alpha(\sigma_{i1}\dts\sigma_{im})\to\alpha(c_1\dts c_m)))
\ee
for given terms $\tau_{ij}$, $\sigma_{il}$. We suppose that we have an
oracle for $HCS(\Phi)$. Denote by $F:=\bigwedge_{i=1}^n
\phi(\tau_{i1}\dts \tau_{ik})$. Let $F_i$ denote $F$ where we
substitute $c_1\mapsto\sigma_{i1}\dts c_m\mapsto\sigma_{im}$ for
$i=1\dts n$. Now we apply our oracle to $F\wedge\bigwedge_{i=1}^n
F_i.$ The terms in this formula may contain constants $c_i$, which are not
in the language of $\Phi$, but we can interpret them as variables to
satisfy the formal definition of $HCS(\Phi)$. Let $A$ be a truth
assignment for the atomic formulas of 
$F\wedge\bigwedge_{i=1}^n F_i$
that makes the formula true. Extend $A$ to an arbitrary assignment that
gives truth values also to those atomic formulas of $\alpha(c_1\dts
c_m)$ and $\alpha(\sigma_{i1}\dts\sigma_{im})$, $i=1\dts n$, for which
$A$ is not defined (e.g., let they be all false). Now we consider two
cases.

1. The assignment $A'$ satisfies the formula (\ref{e2}). Then we are done.

2. Assume it does not. Then, for some $i$, it satisfies
$\alpha(\sigma_{i1}\dts\sigma_{im})$. We define a truth
assignment $A''$ for the formula (\ref{e2}) using the part of $A'$
that assigns values to $F_i$ and
$\sigma_{i1}\dts\sigma_{im}$. Specifically, given $\beta(c_1\dts
c_m)$, an atomic subformula of (\ref{e2}), or an atomic subformula
$\alpha(c_1\dts c_m)$, we assign to it the values that $A'$ gives to
$\beta(\sigma_{i1}\dts \sigma_{im})$. Thus $A''$ satisfies
$\bigwedge_{i=1}^m\phi(\tau_{i1}\dts\tau_{ik})$, because $A'$
satisfies $F_i$, and it also satisfies
$\bigwedge_{i=1}^m(\alpha(\sigma_{i1}\dts\sigma_{im})\to\alpha(c_1\dts
c_m)))$, because $A'$ satisfies
$\alpha(\sigma_{i1}\dts\sigma_{im})$. Thus $A''$ satisfies (\ref{e2}).
\end{proof}

\bpr
  Let $\Phi$ be a consistent sentences in a prenex form $\forall
  x_1\dts x_k\phi(x_1\dts x_k)$ and let $\alpha$ be an an open formula
  with $m$ variables. Suppose that $\Phi\vdash\exists y_1\dts
  y_m\alpha(y_1\dts y_m)$. Then 
$
HCS(\exists \bar{y}\forall \bar{x}(\phi(\bar{x})\wedge\alpha(\bar{y})))
$
is polynomially reducible to
$HCS(\Phi)$.
\epr

\begin{proof}
  The skolemization of the sentence $\exists \bar{y}\forall
  \bar{x}(\phi(\bar{x})\wedge\alpha(\bar{y}))$ is the universal
  sentence
\[
\Psi:=\forall x_1\dts x_k(\phi(x_1\dts x_k)\wedge\alpha(c_1\dts c_m))
\] 
where $c_1\dts c_m$ are new constants. This sentence is provable from
$\Phi\wedge\forall\bar{y}(\alpha(\bar{y})\to\alpha(\bar{c}))$,
because $\Phi$ proves $\exists\bar{y}\ \alpha(\bar{y})$.  Thus,
according to Proposition~\ref{prop4.1}, $HCS(\Psi)$ is polynomially
reducible to 
$HCS(\Phi\wedge\forall\bar{y}(\alpha(\bar{y})\to\alpha(\bar{c})))$.
This problem in turn is reducible to
$\Phi$ by Lemma~\ref{lem4.2}. The polynomial reducibility of
$HCS(\Psi)$, hence also of $HCS(\exists \bar{y}\forall
\bar{x}(\phi(\bar{x})\wedge\alpha(\bar{y})))$, follows by transitivity
of reducibility.
\end{proof}

\section{Provably total search problems in $T$ and HCS($T$)}

Since Herbrand consistency search is defined for sentences, we will
only consider finitely axiomatized theories.  If $T$ contains a
sufficiently strong fragment of arithmetic, or set theory, and $T$ is
sound, we can formalize polynomial time computations in $T$. Then we
can associate with $T$ the class of all search problems that are
provably total in $T$. These are problems that can be defined by a
formula $\rho$ such that $T\vdash\forall x\exists y\rho(x,y)$. In
order to avoid trivialization, we have to restrict the formulas $\rho$ to a
class of formulas that define polynomial time relations in a natural
way. Some theories have symbols for every polynomial time computable
relation, e.g., Cook's $PV$~\cite{cook}. We can also use formulas that
define $\bf NP$ relations, e.g., Buss's $\Sigma_1^b$ in bounded
arithmetic $T_2$~\cite{buss1}. Then the problem of characterizing
provably total search problems is, essentially, equivalent to the
problem of characterizing provable sentences that are universal
closures of $\Sigma_1^b$ formulas. Since $HCS(T)$ is a polynomial
search problem associated with $T$, a natural question arises,
whether or not $HCS(T)$ is in the class of polynomial search problems
provably total in $T$. (Here we assume that $T$ is axiomatized by sentences
in prenex form.) We state this question in a slightly more general
way.
\begin{problem}
Let $T$ be a finitely axiomatized theory, sufficiently strong to be
able to formalize polynomial time computations. Is there a polynomial search
problem $R$ that is provably total in $T$ and such that  $HCS(T)$ is
polynomially reducible to $R$?
\end{problem}

We note that if $T$ is sufficiently strong, then $T$ does not prove
the totality for the natural formalization of $HCS(T)$. Indeed, if $T$
proves Herbrand's theorem, then $T$ can prove that natural provability
(in Hilbert-style calculi, sequent calculi with cuts, etc.) is
equivalent to provability in the sense of Herbrand. Hence the sentence
expressing the totality of $HCS(T)$ is equivalent to the formal
consistency of $T$. Thus by G\"odel's incompleteness theorem, it is
not provable.

A natural approach to solving the problem positively is to try to
express $HCS(T)$ in the following way:
\[
\rho(x,y)\vee\sigma(z),
\]
where $\rho$ is the natural formalization of $HCS(T)$ and $\sigma(z)$
expresses that $z$ is a proof of contradiction from $T$. The task of
this search problem is either to find a satisfying assignment for term
instances of the matrix $\phi$ of $T$ as required by $HCS(T)$, or a
proof of contradiction.  If $T$ is consistent, then this formula is
equivalent to $\rho(x,y)$, hence defines $HCS(T)$. If, moreover, $T$ is
sufficiently strong, then it does prove $\forall x\exists
y,z(\rho(x,y)\vee\sigma(z))$, but this does not suffice. We need $y$
and $z$ to be polynomially bounded:
\[
\forall x\exists y,z(|y|,|z|\leq p(|x|)\wedge(\rho(x,y)\vee\sigma(z))),
\]
for some polynomial $p$. The problem is only to bound $|z|$, since
$|y|$ is polynomially bounded according to the definition of search
problems. Thus we need (to be able to prove from $T$) that if for some
$x$, $\rho(x,y)$ is unsatisfiable, then there exists $z$, a proof of
contradiction from $T$, of at most polynomial length. If $\rho(x,y)$
is unsatisfiable, we know how to construct a contradiction---we have
an unsatisfiable propositional formula, hence we can derive a
contradiction in propositional calculus. However, we do not know if
such a proof can have polynomial length. Thus we do not see how to use
this approach to solve the problem and we tend to conjecture that the
answer is negative.

\section{A relativization}

We conjecture that there is no complete total polynomial search problem. In
order to support this conjecture, we will construct an oracle relative
to which there is no complete polynomial search problem. We will only
prove the proposition for many-one reductions, but the same argument
can surely be used for general reductions.

\bpr There exists an oracle $R$ such that relative to $R$, there is no
complete total polynomial search problem with respect to many-one
polynomial reductions.  
\epr
\bprf
We start by observing that the condition that $R$ has a many-one
polynomial reduction to $S$ can be equivalently defined as follows:
there exist polynomial time computable functions $f(x)$ and $g(x,y)$
such that for all $x$ and $y$
\[
S(f(x),y)\to R(x,g(x,y))
\] 
holds true.

The oracle that we construct will be represented by a ternary relation
$R(p,x,y)$ on binary strings. We will view $p$ as a parameter that
specifies a binary relation $R_p(x,y)$ that may be a total polynomial
search problem. We will construct $R$ so that the condition
\[
R_p(x,y)\to |y|\leq |x|
\]
is satisfied for all $x$ and $y$. Let $\rho$ and $f,g$ be definitions
of a binary relation and two functions by means of polynomial time
oracle Turing machines. Given an oracle $R$, we denote by $\rho_R$ and
$f_R,g_R$ the corresponding relation and functions. We will assume
that the conditions $\rho(x,y)\to |y|\leq |x|$ and $|g(x,y)|\leq
|x|$ are ensured by the definition of $\rho$. We need to construct $R$
so that the following holds true for every $\rho$: \ben
\item either $\rho_R$ is not total, i.e., 
\bel{e4.1}
\exists x\forall y\ \neg\rho_R(x,y),
\ee
\item or for some $p$, $R_p$ is total, but not reducible to $\rho_R$, i.e., for
  every $f,g$,
\bel{e4.2}
\forall x\exists y\ R_p(x,y)\quad\wedge
\ee\bel{e4.3}
\exists x\exists y\ (\rho_R(f_R(x),y)\wedge\neg R_p(x,g_R(x,y))). 
\ee
\een

Our procedure will have two loops---outer and inner. In the outer
loop we go over all definitions $\rho$; in the inner one we go over
all pairs of definitions of $f,g$. In the process we will define $R$
gradually for more and more triples $p,x,y$. At each stage $R$ is
defined only for a finite number of parameters $p$.  At the beginning
of the $i$th outer loop we take $p_i$ such that no value of
$R_{p_i}(x,y)$ has been fixed so far and gradually define
$R_{p_i}(x,y)$. At each stage of this loop $R_{p_i}$ will be defined only
for a finite number of pairs $x,y$.

The outer loop serves us to diagonalize over definitions $\rho$, which
means that at the end of round $i$ the conditions 1. and 2. above will
be satisfied for the $i$th $\rho$.  The partial definition of $R$ will
be denoted by $R^i$.  Similarly, in the inner loop we diagonalize over
functions $f,g$. Let $R^{ij}$ denote the $j$th step of the inner loop
inside of the loop $i$. Then we will get that either 1. holds true and
the loop stops, or 2. holds true for the $j$th pair $f,g$.

Suppose we are in the outer loop $i$. At the beginning of each inner
loop $j$ we assume the following properties of $R^{i(j-1)}$, the oracle
  so far defined.  For every $x$ for which some value of
  $R_{p_i}(x,y)$ has been fixed, there exists some $y$, such that
  $R_{p_i}(x,y')$ (and $|y'|\leq|x|$). Let $n_{ij}$ be a sufficiently
  large number and let $x_{ij}$, $|x_{ij}|=n_{ij}$ be a string such
  that the following conditions are satisfied.  
\bi
\item $R_{p_i}(x_{ij},y)$ has not been fixed for any $y$ and,
\item for every $y'$, $|y'|\leq |f(x_{ij})|$,  the Turing machines of $\rho,f$
  and $g$ cannot query all strings $y$ of length $n_{ij}$ when used on the
  inputs $x_{ij},y',f(x_{ij})$  (because of the polynomial bounds on
  the computations of $\rho,f$ and $g$).
\ei
(The string $x_{ij}$ can be the string of $n_{ij}$ zeros.) First we
extend the oracle so that for every $x$, $|x|<n_{ij}$, there is some
$y$, $|y|\leq |x|$, such that $(x,y)$ is in $R_p$. This is possible,
because we assume that for the strings $x$ used in previous stages
this has already been ensured. We now consider two cases.

\medskip {\it Case 1:} The currently defined oracle can be extended so
that condition~(\ref{e4.3}) is satisfied for $x_{ij}$ and some $y$,
($|y|\leq |x_{ij}|$). In this case we fix the minimum number of values
of $R$ that are needed to ensure this condition. Then it is still
possible to add pairs $(x,y)$ to $R_p$ to ensure $\exists y\ R_p(x,y)$
for all strings $x$ so far used.

\medskip
{\it Case 2:} The opposite is true. This means that for every extension of
the so far specified oracle $R$, and every $y$, $|y|\leq |x_{ij}|$,
the implication 
\[
\rho_R(f_R(x_{ij}),y)\to R_{p_i}(x_{ij},g_R(x_{ij},y))
\]
is satisfied. In particular, the implication will be satisfied if we
fix $R$ so that for all  $z$, $|z|\leq |x_{ij}|$, $\neg
R_{p_i}(x_{ij},z)$. It follows that $\neg\rho_R(f_R(x_{ij}),y)$ for
all $y$, $|y|\leq |f_R(x_{ij})|$. Hence (\ref{e4.1}) holds true for
all further extension of so far defined~$R$.
\eprf

\section{Conclusions}

We still do not fully understand the relation between provability and
polynomial reducibility of the corresponding Herbrand consistency
search problems. The most important problem is:

\begin{problem}\label{problem1}
Let $\Phi$ and $\Psi$ be consistent sentences in prenex forms. Suppose
that $\Phi\vdash\Psi$. Is then $HCS(\Psi)$ polynomially reducible to
$HCS(\Phi)$? 
\end{problem}

We have only been able to solve the problem in two special cases: for
universal sentences $\Psi$ and for existential sentences $\Psi$. If
the answer is negative, then the concept of Herbrand consistency is
not well-behaved. In such a case it would be better to compare the 
provability of $\Phi\to\Psi$ with the reducibility of $HCS(\Psi)$ to
$HCS(\alpha)$ for all prenex sentences $\alpha$ derivable from~$\Phi$.

Another interesting open problem is whether or not polynomial reducibility of
$HCS(\Psi)$ to $HCS(\Phi)$ implies $\Phi\vdash\Psi$ at least in some
special cases.

\subsection*{Acknowledgment}
I would like to thank Neil Thapen for his comments on an early draft of
this paper.

\end{document}